\documentclass[10pt, conference, letterpaper]{IEEEtran}
\IEEEoverridecommandlockouts

\usepackage[letterpaper, top=0.75in, bottom=1.05in, left=0.63in, right=0.63in, columnsep=0.22in]{geometry}
\usepackage{amsmath,amsfonts,amssymb,amsthm,bm}
\usepackage{graphicx}
\usepackage{cite}
\usepackage{algorithm}
\usepackage{algorithmic}
\usepackage{booktabs}

\newtheorem{lemma}{Lemma}
\newtheorem{theorem}{Theorem}
\newtheorem{remark}{Remark}

\begin{document}

\title{Physics-Aware Power Optimization for Rydberg Quantum Arrays via Barankin Bound Minimization}

\author{Yixiao~Li,
        Jiayi~Zhang,~\IEEEmembership{Senior~Member,~IEEE,}
        Bokai~Xu,
        and Yin~Zhang%
\thanks{This work was supported in part by the Fundamental Research Funds for the Central Universities under Grant 2025YJS017. Y. Li, J. Zhang, B. Xu, and Y. Zhang are with the State Key Laboratory of Advanced Rail Autonomous Operation and the School of Electronic and Information Engineering, Beijing Jiaotong University. (Corresponding authors: Jiayi Zhang; Bokai Xu.)
\textit{}}%
}

\maketitle

\begin{abstract}
Rydberg-atom quantum uniform linear arrays (RAQ-ULAs) offer a promising sensing architecture for direction-of-arrival (DOA) estimation in terahertz (THz) beam alignment. Existing studies often model the quantum receiver as a macroscopic linear block and rely on the Cramer-Rao bound (CRB) for array evaluation or power allocation. However, as a local variance bound, the CRB cannot capture threshold breakdown caused by spatial ambiguities in low-signal-to-noise-ratio (SNR) regimes. This paper proposes a physics-aware power optimization framework for RAQ-ULAs based on stochastic Barankin bound minimization. We derive a closed-form stochastic multipoint Barankin bound (BRB) matrix under a low-SNR integrability condition and introduce an arcsine-prior calibration to characterize bounded-domain error saturation. By incorporating a Lindblad-guided electromagnetically induced transparency (EIT) readout model, we couple the ambiguity-sensitive BRB with photon shot noise, power broadening, and laser Rabi frequencies. A CRB-optimized analytical baseline is further derived to expose the limitation of local-SNR-based allocation. To solve the resulting nonconvex problem, we develop a backtracking majorization-minimization (MM) algorithm with a Lipschitz-based quadratic surrogate. The algorithm ensures monotonic decrease and converges to a feasible stationary point. Simulations show that the proposed framework predicts threshold breakdown more accurately than the CRB and enlarges the reliable operating region under severe THz attenuation.
\end{abstract}

\section{Introduction} \label{sec:intro}

Emerging sixth-generation (6G) communication systems are expected to leverage millimeter-wave (mmWave) and terahertz (THz) bands to support ultra-high data rates \cite{Ai_2025_6G}. However, high-frequency propagation is subject to severe free-space path loss and molecular absorption, which makes accurate beam alignment critical for initial access \cite{Heath_2016}. In the initial access phase of THz communications, omnidirectional or quasi-omnidirectional antenna patterns provide negligible beamforming gain, which, combined with severe propagation attenuation, induces an acute link-budget deficit. Consequently, the receiver operates in an extremely low-signal-to-noise-ratio (SNR) regime where directional beam alignment must be accomplished before high-gain directional data transmission can begin. At THz frequencies, conventional metallic receivers are further limited by thermal noise, hardware impairments, and sensitivity constraints. Rydberg-atom-based quantum receivers have therefore emerged as a promising alternative, exploiting electromagnetically induced transparency (EIT) and Autler-Townes (AT) splitting to transduce radio-frequency (RF) electric fields with high sensitivity. These advantages motivate the use of Rydberg-atom quantum uniform linear arrays (RAQ-ULAs) for direction-of-arrival (DOA) estimation, where the large transition dipole moments of Rydberg states provide a sensitive spatial-sensing mechanism \cite{Gong_RAQR_2026, Liu_2022, Richardson_2025}.

Despite these advances, existing RAQ-ULA studies remain limited in both physical modeling and performance characterization. First, many macroscopic array models represent the quantum receiver as a linear block with fixed equivalent gain and noise variance \cite{Meyer_2021}. Although this abstraction enables tractable array-level analysis, it overlooks that the macroscopic conversion gain and optical photon shot noise (PSN) are coupled through the probe and coupling laser Rabi frequencies. In practical EIT readout, increasing the probe intensity can strengthen the optical response, but it also raises the shot-noise floor and induces power broadening, thereby degrading the transparency window \cite{Fleischhauer_2005}. Therefore, power allocation in RAQ-ULAs cannot be accurately characterized by treating the quantum gain and noise variance as independent constants.

Second, array performance is commonly characterized by the Cramer-Rao bound (CRB) \cite{Krim_1996}. The CRB is derived from the Fisher information matrix (FIM) and quantifies the local estimation-error variance in a neighborhood of the true parameter. However, in low-SNR THz initial access, DOA estimators frequently experience threshold breakdown, where the estimate escapes from the main lobe and locks onto finite-aperture ambiguity peaks \cite{VanTrees_2002}. Because the CRB only reflects the local curvature of the likelihood function, it belongs to local variance bounds and inherently cannot capture such nonlocal ambiguity errors, yielding overly optimistic predictions in severely attenuated regimes. In contrast, the Barankin bound (BRB) belongs to the family of large-error global performance bounds that explicitly incorporate finite-separation test points across the array manifold, thereby capturing spatial ambiguities and predicting the threshold breakdown point. Therefore, characterizing bounded-domain error saturation requires a global performance bound with prior calibration consistent with the RAQ-ULA spatial-frequency geometry.

Motivated by these limitations, this paper proposes a physics-aware power optimization framework for RAQ-ULAs based on stochastic BRB minimization. First, we derive a closed-form stochastic multipoint BRB matrix under a low-SNR integrability condition and introduce an arcsine-prior calibration to characterize bounded-domain error saturation. Second, we develop a physics-aware mapping from EIT readout dynamics to the array covariance model, coupling the quantum conversion gain, PSN, power broadening, and laser Rabi frequencies. Based on this mapping, we derive a CRB-optimized baseline to expose the limitation of local-SNR-based optimization in threshold-dominated regimes. Finally, we formulate a nonconvex power allocation problem that minimizes the prior-calibrated BRB under physical laser-power constraints, and solve it using a backtracking majorization-minimization (MM) algorithm with a Lipschitz-based quadratic surrogate. Simulations validate the trade-off between power broadening and PSN, show that the BRB predicts threshold breakdown more accurately than the CRB, and demonstrate that the proposed framework enlarges the reliable operating region under severe THz attenuation.

\section{THz Channel and Stochastic System Model} \label{sec:system}

\begin{figure}[t]
\centering
\includegraphics[width=\columnwidth]{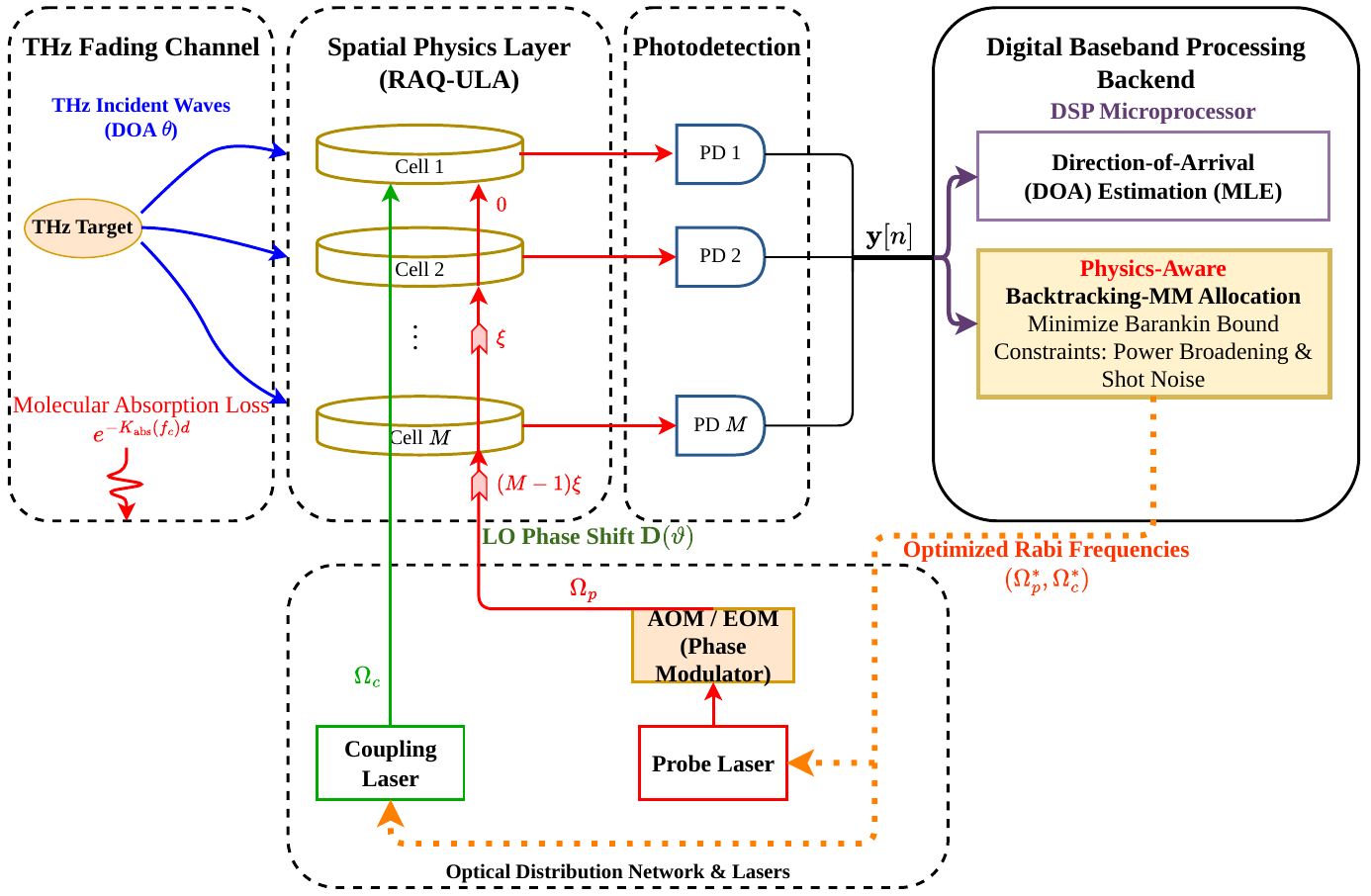}
\caption{System architecture of the proposed RAQ-ULA for 6G beam alignment.}
\label{fig:system_model}
\end{figure}

This section develops the observation model for the RAQ-ULA. As illustrated in Fig.~\ref{fig:system_model}, we first characterize the THz line-of-sight (LoS) channel, then formulate the array geometry under optical phase modulation, and finally derive the baseband stochastic model. This construction links the array covariance matrix to the underlying quantum parameters.

Consider an $M$-element RAQ-ULA receiving a narrowband THz signal with carrier frequency $f_c$ and wavelength $\lambda_c$ from a target at distance $d$ and DOA $\theta$. Although ultra-broadband transmission in THz systems can induce spatial-frequency dispersion (beam squint) across wide frequency bands, initial access and beam alignment typically employ narrow sub-band training pilots or central orthogonal frequency division multiplexing (OFDM) subcarriers \cite{Wang_BeamSquint_2019} within the channel coherence interval, which renders the narrowband array model well-grounded. THz propagation is subject to severe free-space path loss and molecular absorption. According to the Beer-Lambert law, the deterministic LoS channel power gain is modeled as
\begin{equation} \label{eq:thz_channel}
|h(f_c, d)|^2 = \left(\frac{c}{4\pi d f_c}\right)^2 e^{-K_{\text{abs}}(f_c) d},
\end{equation}
where $c$ is the speed of light and $K_{\text{abs}}(f_c)$ denotes the molecular absorption coefficient. In the THz band, $K_{\text{abs}}(f_c)$ depends on environmental factors such as water-vapor partial pressure and ambient temperature, resulting in frequency-selective absorption peaks. Since non-line-of-sight (NLoS) components are typically strongly attenuated by the limited diffraction and scattering of submillimeter waves, the channel is assumed to be LoS-dominant. Although future ultra-massive THz systems may require spherical-wave models in the near-field regime \cite{Xu_NearField_2024}, the compact aperture of the RAQ-ULA supports a far-field approximation. Accordingly, the spatial channel is modeled as rank-one. The received RF signal variance is then given by
$P_s = P_{\text{tx}} G_{\text{tx}} G_{\text{rx}} |h(f_c, d)|^2$, where $P_{\text{tx}}$, $G_{\text{tx}}$, and $G_{\text{rx}}$ denote the transmit power, transmit antenna gain, and receive antenna gain, respectively. Such severe attenuation leads to low-SNR operating conditions, motivating the use of high-sensitivity quantum arrays for reliable beam alignment.

The RAQ-ULA transduces the incident THz electric field into a measurable baseband signal through Rydberg atomic transitions. This process relies on a four-level atomic ladder scheme, typically implemented with alkali atoms such as $^{85}$Rb or $^{133}$Cs confined in localized vapor cells. A probe laser and a coupling laser drive the atoms toward highly excited Rydberg states. When the incident THz wave is resonant with the transition between two adjacent Rydberg states, it induces AT splitting in the EIT spectrum. Since the transition dipole moment of Rydberg atoms scales approximately as $n^2$, where $n$ is the principal quantum number, the receiver can achieve high electric-field sensitivity and broad frequency tunability \cite{xu2025channelestimation}.

From an engineering perspective, current room-temperature Rydberg vapor-cell receivers typically exhibit instantaneous baseband bandwidths on the order of $1$--$10$~MHz, limited by atomic transit times and optical decay rates \cite{Meyer_Wideband_2020}. Furthermore, while the half-wavelength inter-element spacing $\Delta d = \lambda_c/2$ at THz frequencies ($0.1$--$1$~THz corresponds to $\lambda_c \approx 3$--$0.3$~mm) poses substantial fabrication challenges for conventional glass vapor cells, emerging micro-machined atomic vapor cells and integrated photonic waveguides provide viable pathways toward sub-millimeter array integration, or alternatively, sparse array configurations can be utilized to accommodate physical cell dimensions without sacrificing angular resolution \cite{Meyer_Wideband_2020}.

Assume that the vapor cells are uniformly spaced with inter-element spacing $\Delta d$. The spatial steering vector $\mathbf{a}(\theta) \in \mathbb{C}^{M \times 1}$ is defined as
\begin{equation} \label{eq:steering_vector}
\mathbf{a}(\theta) = \left[ {1}, e^{j\mu_s}, \dots, e^{j(M-1)\mu_s} \right]^T,
\end{equation}
where $\mu_s \triangleq \frac{2\pi \Delta d}{\lambda_c} \sin\theta$ denotes the RF spatial frequency. In contrast to conventional metallic arrays, the baseband phase of the RAQ-ULA is also affected by the local oscillator (LO) laser system. When the optical LO network introduces a fixed incidence angle $\vartheta$, the deterministic LO phase-shift matrix $\mathbf{D}(\vartheta) \in \mathbb{C}^{M \times M}$ is given by
\begin{equation} \label{eq:lo_matrix}
\mathbf{D}(\vartheta) = \operatorname{diag}\left\{ {1}, e^{-j\xi}, \dots, e^{-j(M-1)\xi} \right\},
\end{equation}
where $\xi \triangleq \frac{2\pi \Delta d}{\lambda_{\text{LO}}} \sin\vartheta$ and $\lambda_{\text{LO}}$ denotes the effective LO wavelength. The matrix $\mathbf{D}(\vartheta)$ models the optical path delay distributed across vapor cells via polarization-maintaining fibers or free-space splitters. The resulting composite steering vector is written as
$\mathbf{b}(\theta) \triangleq \mathbf{D}(\vartheta)\mathbf{a}(\theta)$.

Following the stochastic array processing model for fluctuating sources \cite{Stoica_1990}, the incident signal $s[n]$ at the $n$-th snapshot is modeled as a zero-mean circularly symmetric complex Gaussian (CSCG) random variable with variance $P_s$. The snapshots are assumed to be independent and identically distributed within the channel coherence interval. The received baseband signal vector $\mathbf{y}[n] \in \mathbb{C}^{M \times 1}$ is expressed as
\begin{equation} \label{eq:signal_model}
\mathbf{y}[n] = \beta_0 \mathbf{b}(\theta) s[n] + \mathbf{w}[n], \quad n=1,\dots,N,
\end{equation}
where $\beta_{0} \triangleq \sqrt{E_{e}}\rho\Phi$ denotes the macroscopic scalar gain. Here, $E_e$ is the electric-field normalization factor, $\rho$ is the microscopic quantum conversion gain, and $\Phi$ is the field-to-voltage conversion coefficient. The additive noise vector follows $\mathbf{w}[n] \sim \mathcal{CN}(\mathbf{0},\sigma^2 \mathbf{I}_M)$, where the spatially uncorrelated covariance reflects independent photodetector readouts across vapor cells. In balanced photodetector architectures, differential detection suppresses classical laser intensity noise, leaving $\sigma^2$ governed by thermal background noise and optical PSN. As shown in Section~\ref{sec:optimization}, both $\rho$ and $\sigma^2$ depend nonlinearly on the laser Rabi frequencies. Consequently, $\mathbf{y}[n] \sim \mathcal{CN}(\mathbf{0},\mathbf{R}_{\mathbf{y}}(\theta))$, where the spatial covariance matrix is
\begin{equation} \label{eq:covariance}
\mathbf{R}_{\mathbf{y}}(\theta) = |\beta_0|^2 P_s \mathbf{b}(\theta)\mathbf{b}^H(\theta) + \sigma^2 \mathbf{I}_M.
\end{equation}

For beam-alignment performance evaluation, the normalized receive beamforming gain $G_{\mathrm{BF}}(\Delta \theta)$ resulting from the DOA estimation error $\Delta \theta = \hat{\theta} - \theta$ can be approximated by the Dirichlet kernel \cite{Heath_2016}. Under the small-error approximation $\sin(\theta+\Delta\theta)-\sin\theta \approx \cos(\theta)\Delta\theta$, it is given by
\begin{equation} \label{eq:gbf}
G_{\mathrm{BF}}(\Delta \theta) \approx \left| \frac{\sin\big(M \pi \frac{\Delta d}{\lambda_c} \cos(\theta) \Delta \theta\big)}{M \sin\big(\pi \frac{\Delta d}{\lambda_c} \cos(\theta) \Delta \theta\big)} \right|^2.
\end{equation}

\section{Stochastic Performance Bounds} \label{sec:bounds}

This section develops stochastic performance bounds for DOA estimation. We first use the stochastic CRB as a local benchmark. To characterize spatial ambiguities and the resulting threshold breakdown in THz initial access, we then derive a closed-form stochastic multipoint BRB under the low-SNR integrability condition. The BRB provides a non-asymptotic characterization of estimation error by accounting for finite-separation uncertainty that cannot be captured by local differential analysis. By applying the Sherman-Morrison matrix inversion lemma to \eqref{eq:covariance}, the inverse covariance matrix is given by
\begin{equation} \label{eq:invR}
\mathbf{R}_{\mathbf{y}}^{-1}(\theta) = \sigma^{-2}\mathbf{I}_M - \nu \mathbf{b}(\theta)\mathbf{b}^H(\theta),
\end{equation}
where the scalar coefficient $\nu$ is defined as
\begin{equation} \label{eq:nu}
\nu \triangleq \frac{|\beta_0|^2 P_s}{\sigma^2(\sigma^2 + M |\beta_0|^2 P_s)}.
\end{equation}
The determinant of the spatial covariance matrix is
$\det(\mathbf{R}_{\mathbf{y}}(\theta)) = \sigma^{2(M-1)}(\sigma^2 + M |\beta_0|^2 P_s)$.

\subsection{Stochastic CRB}

For the zero-mean CSCG observation model, the stochastic CRB for DOA estimation follows from the Slepian--Bangs formula as
\begin{equation} \label{eq:crb}
\mathrm{CRB}_{\mathrm{sto}}(\theta) =
\frac{1}{N \operatorname{tr}\left[
\left(
\mathbf{R}_{\mathbf{y}}^{-1}(\theta)
\frac{\partial \mathbf{R}_{\mathbf{y}}(\theta)}{\partial \theta}
\right)^2
\right]}.
\end{equation}
Substituting \eqref{eq:invR} into \eqref{eq:crb} shows that the Fisher information is governed by the local derivative metric $\dot{\mathbf{b}}^H \boldsymbol{\Pi}_{\mathbf{b}}^{\perp} \dot{\mathbf{b}}$, where $\dot{\mathbf{b}} \triangleq \partial \mathbf{b}(\theta)/\partial\theta$ and $\boldsymbol{\Pi}_{\mathbf{b}}^{\perp} \triangleq \mathbf{I}_M - \frac{1}{M}\mathbf{b}\mathbf{b}^H$ is the orthogonal complement projector. Under the ULA geometry, this metric evaluates to $\dot{\mathbf{b}}^H \boldsymbol{\Pi}_{\mathbf{b}}^{\perp} \dot{\mathbf{b}} = \frac{\pi^2 \Delta d^2}{\lambda_c^2}\frac{M(M^2-1)}{3}\cos^2\theta$, revealing that the local Fisher information scales as $\mathcal{O}(M^3)$. However, this metric captures only local likelihood curvature in the vicinity of the true DOA and ignores finite-separation ambiguity errors across the array sidelobe pattern. Consequently, the CRB yields overly optimistic predictions in attenuated THz scenarios where ambiguity peaks dominate.

\subsection{Stochastic Multipoint Barankin Bound}

To capture threshold breakdown induced by spatial ambiguities, we adopt the Hammersley--Chapman--Robbins formulation with the test-point vector $\mathbf{t} = [\theta^{(1)}-\theta, \dots, \theta^{(L)}-\theta]^T$ \cite{Chaumette_2008}. The test points are selected near the local ambiguity peaks of the finite-aperture RAQ-ULA correlation function. Specifically, candidate points are generated as
\begin{equation} \label{eq:test_points}
\tilde{\theta}^{(l,\pm)} = \arcsin\left( \sin\theta \pm \frac{(2l+1) \lambda_c}{2M \Delta d} \right),
\end{equation}
for $l = 1, 2, \dots, L_0$. These candidate positions correspond to the null-to-null beamwidth intervals and sidelobe crests of the array Dirichlet response. Only candidates satisfying
$\big\vert{} \sin\theta \pm \frac{(2l+1) \lambda_c}{2M \Delta d} \big\vert{} \le 1$
are retained and re-indexed as $\{\theta^{(l)}\}_{l=1}^{L}$.

It is worth noting that the test point selection strategy in \eqref{eq:test_points} relies on the far-field Dirichlet kernel. For extremely large-scale arrays where near-field spherical wave effects become prominent in 6G systems \cite{Liu_NearField_2024, Lu_NearField_2022, Liu_CellFree_2024}, the spatial ambiguity function will expand into the distance domain, requiring a 2D angle-distance test point grid.

\begin{lemma}[Stochastic Barankin Matrix] \label{lemma:brb_mat}
For the RAQ-ULA under the CSCG observation model, if $\mathbf{M}_{lr}\succ \mathbf{0}$ for all retained test-point pairs, the entries of the Barankin matrix $\boldsymbol{\mathcal{B}}(\theta) \in \mathbb{R}^{L \times L}$ are given by
\begin{equation} \label{eq:brb_element}
[\boldsymbol{\mathcal{B}}(\theta)]_{lr}
=
\left(
\frac{\det(\mathbf{R}_{\mathbf{y}}(\theta))}
{\det(\mathbf{R}_{\mathbf{y}}^{(l)})\det(\mathbf{R}_{\mathbf{y}}^{(r)})\det(\mathbf{M}_{lr})}
\right)^N,
\end{equation}
where $\mathbf{R}_{\mathbf{y}}^{(l)} \triangleq \mathbf{R}_{\mathbf{y}}(\theta^{(l)})$, and
\begin{equation}
\mathbf{M}_{lr}
\triangleq
(\mathbf{R}_{\mathbf{y}}^{(l)})^{-1}
+
(\mathbf{R}_{\mathbf{y}}^{(r)})^{-1}
-
\mathbf{R}_{\mathbf{y}}^{-1}(\theta).
\end{equation}
\end{lemma}

\begin{IEEEproof}
Let $p(\mathbf{y}\vert{}\theta)$ denote the one-snapshot CSCG probability density function. By definition, the $(l,r)$-th entry of $\boldsymbol{\mathcal{B}}(\theta)$ is
\begin{equation} \label{eq:expectation}
\mathbb{E}_{\theta}\left[
\frac{p(\mathbf{y}|\theta^{(l)})p(\mathbf{y}|\theta^{(r)})}
{p^2(\mathbf{y}|\theta)}
\right]
=
\int_{\mathbb{C}^M}
\frac{p(\mathbf{y}|\theta^{(l)})p(\mathbf{y}|\theta^{(r)})}
{p(\mathbf{y}|\theta)}
\mathrm{d}\mathbf{y}.
\end{equation}
Substituting the CSCG densities and collecting the exponential terms, the exponent reduces to $-\mathbf{y}^H\mathbf{M}_{lr}\mathbf{y}$. Hence, \eqref{eq:expectation} becomes
\begin{equation} \label{eq:integral}
\frac{\det(\mathbf{R}_{\mathbf{y}}(\theta))}
{\pi^M \det(\mathbf{R}_{\mathbf{y}}^{(l)})\det(\mathbf{R}_{\mathbf{y}}^{(r)})}
\int_{\mathbb{C}^M}
\exp\left(-\mathbf{y}^H\mathbf{M}_{lr}\mathbf{y}\right)
\mathrm{d}\mathbf{y}.
\end{equation}
The integral is finite when $\mathbf{M}_{lr}\succ\mathbf{0}$. Substituting \eqref{eq:invR} into \eqref{eq:mlr_expansion} gives
\begin{equation} \label{eq:mlr_expansion}
\mathbf{M}_{lr} = \sigma^{-2}\mathbf{I}_M + \nu \mathbf{b}\mathbf{b}^H - \nu (\mathbf{b}_l\mathbf{b}_l^H + \mathbf{b}_r\mathbf{b}_r^H),
\end{equation}
where $\mathbf{b}\triangleq\mathbf{b}(\theta)$ and $\mathbf{b}_i\triangleq\mathbf{b}(\theta^{(i)})$. Since $\nu \mathbf{b}\mathbf{b}^H \succeq \mathbf{0}$, Weyl's inequality gives
$\lambda_{\min}(\mathbf{M}_{lr}) \ge \sigma^{-2} - \nu \lambda_{\max}(\mathbf{b}_l\mathbf{b}_l^H + \mathbf{b}_r\mathbf{b}_r^H)$.
Using $\Vert{}\mathbf{b}_l\Vert{}_2^2=\Vert{}\mathbf{b}_r\Vert{}_2^2=M$, we have
$\lambda_{\max}(\mathbf{b}_l\mathbf{b}_l^H + \mathbf{b}_r\mathbf{b}_r^H) \le 2M$.
Thus, $\mathbf{M}_{lr}\succ\mathbf{0}$ is guaranteed by
$\sigma^{-2}-2M\nu>0$, which is equivalent to
\begin{equation} \label{eq:convergence_cond}
\sigma^2 > M |\beta_0|^2 P_s.
\end{equation}
This condition agrees with the low-SNR THz initial-access regime considered in this work. Although the mathematical form of $\mathbf{M}_{lr}$ in \eqref{eq:mlr_expansion} and the condition in \eqref{eq:convergence_cond} resemble those of classical ULA processing, the fundamental distinction in RAQ-ULAs is that the scalar gain $\beta_0 = \sqrt{E_e}\rho\Phi$ and noise variance $\sigma^2 = N_{\mathrm{bg}} + c_{\mathrm{shot}}\Omega_p^2$ are physically governed by the microscopic quantum EIT readout dynamics and laser Rabi frequencies $(\Omega_p, \Omega_c)$. In classical arrays, gain and noise are fixed constants, making performance bounds passive evaluation tools; in RAQ-ULAs, the bound is explicitly coupled to laser powers, enabling active quantum-state optimization to reshape the ambiguity landscape. The integral then evaluates to $\pi^M/\det(\mathbf{M}_{lr})$, and extending the result to $N$ independent snapshots gives \eqref{eq:brb_element}.
\end{IEEEproof}

To account for the bounded DOA domain $[-\pi/2, \pi/2]$, we incorporate a prior-based calibration into the BRB \cite{Renaux_Prior_2008}. In the pure-noise regime, the estimator has no reliable directional information. Since the RAQ-ULA spatial-frequency geometry is naturally represented in the $\sin\theta$ domain, an uninformative estimate in this domain induces an arcsine-transformed prior for $\theta$. The corresponding maximum uncertainty is
$V_{\mathrm{arc}} = \pi^2/4 - 2$.
The prior-calibrated stochastic BRB is then defined as
\begin{equation} \label{eq:cbrb}
\begin{split}
\mathrm{BRB}_{\mathrm{sto},\delta}^{c}(\theta) &= \bigg( \Big( \mathbf{t}^T \big(\boldsymbol{\mathcal{B}}(\theta) - \mathbf{1}_L\mathbf{1}_L^T \\
&\quad + \delta \mathbf{I}_L \big)^{-1} \mathbf{t} \Big)^{-1} + V_{\mathrm{arc}}^{-1} \bigg)^{-1},
\end{split}
\end{equation}
where $\delta>0$ is a small regularization constant for maintaining invertibility in ill-conditioned regimes.

\begin{remark}[Threshold Sensitivity] \label{rem:brb_threshold}
Unlike the local CRB, the BRB depends on finite-separation likelihood ratios through $\mathbf{M}_{lr}$, thereby capturing spatial-ambiguity effects. When the signal power drops, random noise fluctuations cause the likelihood values at sidelobe test points to compete with the mainlobe, triggering threshold breakdown. Since $\beta_0$ incorporates the quantum conversion gain, the BRB links threshold behavior directly to microscopic laser parameters, providing an actionable metric for proactive array power optimization. As the incident SNR decreases, the prior-calibrated BRB approaches the saturation level $V_{\mathrm{arc}}$, providing a more reliable characterization of threshold breakdown than the local CRB.
\end{remark}

\section{Physics-Aware Power Optimization} \label{sec:optimization}

Section~\ref{sec:bounds} shows that mitigating threshold breakdown requires a global performance criterion beyond the local CRB. This section formulates a physics-aware laser-power allocation problem by mapping microscopic quantum dynamics to macroscopic array parameters, enabling laser-power variables to jointly shape the array covariance and the ambiguity-sensitive BRB. We then develop an MM framework to optimize the resulting nonconvex BRB objective.

\subsection{Quantum Macroscopic Mapping and CRB Baseline}

The multilevel Rydberg atomic dynamics are governed by the Lindblad master equation \cite{Meyer_2021}. Under two-photon resonance, i.e., $\Delta_p=0$ and $\Delta_c=0$, the equivalent quantum conversion gain $\rho$ and noise variance $\sigma^2$ can be approximated as coupled functions of the probe and coupling laser Rabi frequencies, $\Omega_p$ and $\Omega_c$, respectively:
\begin{equation} \label{eq:gain_eit}
\rho(\Omega_p, \Omega_c) \approx \frac{\kappa \Omega_p \Omega_c}{\Omega_c^2 + \Gamma_{\mathrm{EIT}}^2 + \Omega_p^2}, \quad
\sigma^2(\Omega_p) \approx N_{\mathrm{bg}} + c_{\mathrm{shot}} \Omega_p^2.
\end{equation}
The term $\Omega_p^2$ in the denominator captures probe-induced power broadening. Specifically, an excessively strong probe field perturbs atomic coherence and the dark state, narrowing the transparency window and degrading the optical response. Meanwhile, increasing $\Omega_p$ raises the PSN floor, where $c_{\mathrm{shot}}$ is the shot-noise coefficient and $N_{\mathrm{bg}}$ denotes the background thermal noise. Therefore, laser-power allocation must balance quantum conversion gain, power broadening, and PSN.

Under conventional CRB-based optimization, the design reduces to maximizing the effective array quantum SNR metric $\eta = |\beta_0|^2/\sigma^2 \propto \rho^2/\sigma^2$. Let $x=\Omega_p^2$ and $y=\Omega_c^2$. Ignoring constant scaling factors, the objective becomes
\begin{equation} \label{eq:snr_metric}
\eta(x, y) =
\frac{x y}{(y + \Gamma_{\mathrm{EIT}}^2 + x)^2 (N_{\mathrm{bg}} + c_{\mathrm{shot}}x)}.
\end{equation}

\begin{theorem}[CRB-Optimized Analytical Baseline] \label{thm:optimal_power}
For any fixed coupling intensity $y>0$, the CRB-based SNR metric in \eqref{eq:snr_metric} has a unique maximizer over $x>0$, given by
\begin{equation} \label{eq:closed_form_x}
x^*
=
\frac{
-N_{\mathrm{bg}}
+
\sqrt{
N_{\mathrm{bg}}^2
+
8c_{\mathrm{shot}}N_{\mathrm{bg}}
\left(y+\Gamma_{\mathrm{EIT}}^2\right)
}
}
{4c_{\mathrm{shot}}}.
\end{equation}
\end{theorem}

\begin{IEEEproof}
Taking the logarithm of $\eta(x,y)$ and differentiating with respect to $x$ gives
\begin{equation} \label{eq:derivative}
\frac{\partial \ln \eta(x,y)}{\partial x}
=
\frac{1}{x}
-
\frac{2}{y+\Gamma_{\mathrm{EIT}}^2+x}
-
\frac{c_{\mathrm{shot}}}{N_{\mathrm{bg}}+c_{\mathrm{shot}}x}.
\end{equation}
Setting \eqref{eq:derivative} to zero and clearing the denominators yields
$2c_{\mathrm{shot}}x^2+N_{\mathrm{bg}}x-(y+\Gamma_{\mathrm{EIT}}^2)N_{\mathrm{bg}}=0$.
Since all parameters are positive, Descartes' rule of signs implies exactly one positive root. Moreover, $\eta(x,y)\to0$ as $x\to0^+$ and $x\to\infty$, so this stationary point is the unique maximizer over $x>0$.
\end{IEEEproof}

Although Theorem~\ref{thm:optimal_power} gives the optimal allocation for local-SNR maximization, it ignores finite-separation ambiguity errors. At extremely low RF power $P_s$, the SNR metric $\eta(x,y)$ may still cause premature threshold breakdown because the ambiguity amplification in \eqref{eq:brb_element} across sidelobes is not considered. The CRB-optimized solution therefore serves as an analytical baseline rather than the final design criterion in ambiguity-limited regimes.

\begin{algorithm}[t]
\caption{BRB Minimization via Backtracking MM} \label{alg:mm}
\begin{algorithmic}[{1}]
\REQUIRE $\mathcal{F}$, $\mathbf{u}^{(0)}\in\mathcal{F}$, $L_0>0$, $\tau>1$, $\epsilon_{\mathrm{tol}}$.
\ENSURE $(\Omega_p^*,\Omega_c^*)$.
\STATE Initialize $k=0$.
\REPEAT
    \STATE Evaluate $\nabla J(\mathbf{u}^{(k)})$ and set $L_k=L_0$.
    \LOOP
        \STATE $\mathbf{z}=\mathcal{P}_{\mathcal{F}}\!\left(\mathbf{u}^{(k)}-\frac{1}{L_k}\nabla J(\mathbf{u}^{(k)})\right)$.
        \IF{$J(\mathbf{z}) \le g_{L_k}(\mathbf{z}\mid\mathbf{u}^{(k)})$}
            \STATE \textbf{break}
        \ELSE
            \STATE $L_k \leftarrow \tau L_k$.
        \ENDIF
    \ENDLOOP
    \STATE $\mathbf{u}^{(k+1)} \leftarrow \mathbf{z}$, $k \leftarrow k+1$.
\UNTIL{$\|\mathbf{u}^{(k)}-\mathbf{u}^{(k-1)}\|_2 \le \epsilon_{\mathrm{tol}}$.}
\STATE \textbf{return} $\Omega_p^*=\sqrt{[\mathbf{u}^{(k)}]_1}$,
\STATE \hspace{1.6em} $\Omega_c^*=\sqrt{[\mathbf{u}^{(k)}]_2}$.
\end{algorithmic}
\end{algorithm}

\subsection{BRB Minimization via Lipschitz-MM Framework}

To enlarge the reliable operating region, we formulate the laser-power allocation problem by minimizing the prior-calibrated BRB in \eqref{eq:cbrb}. Defining $\mathbf{u}=[x,y]^T$, the objective is
$J(\mathbf{u}) \triangleq \mathrm{BRB}_{\mathrm{sto},\delta}^{c}(\mathbf{u})$.
To keep the EIT readout within its valid operating regime, the probe intensity is constrained relative to the coupling intensity. A small lower bound $\varepsilon>0$ is introduced to avoid degenerate zero-power points. Given the hardware power budget $P_{\mathrm{total}}$ and the EIT boundary coefficient $\gamma\in(0,1)$, the feasible set is the closed convex polytope
\begin{equation} \label{eq:feasible}
\begin{aligned}
\mathcal{F}=\{\,\mathbf{u}\in\mathbb{R}^2 \mid&
x\ge\varepsilon,\; y\ge\varepsilon,\\
&x+y\le P_{\mathrm{total}},\; x\le\gamma y\,\}.
\end{aligned}
\end{equation}
This set is nonempty when $0<\varepsilon<\frac{\gamma}{1+\gamma}P_{\mathrm{total}}$.

Minimizing $J(\mathbf{u})$ over $\mathcal{F}$ is nonconvex due to the nonlinear EIT mapping and the determinant-ratio structure of the BRB. The regularization term $\delta\mathbf{I}_L$ in \eqref{eq:cbrb} improves numerical stability by avoiding singular matrix inversion in ill-conditioned regimes. Since the regularized objective is smooth over the compact set $\mathcal{F}$, $\nabla J(\mathbf{u})$ is Lipschitz continuous with some finite constant $L>0$. By the descent lemma, we construct the quadratic surrogate
\begin{equation} \label{eq:surrogate}
\begin{aligned}
g_{L_k}(\mathbf{u}\mid\mathbf{u}^{(k)})
&= J(\mathbf{u}^{(k)})
+ \nabla J(\mathbf{u}^{(k)})^T
(\mathbf{u}-\mathbf{u}^{(k)}) \\
&\quad + \frac{L_k}{2}
\|\mathbf{u}-\mathbf{u}^{(k)}\|_2^2.
\end{aligned}
\end{equation}

For any $L_k \ge L$, $g_{L_k}(\mathbf{u} \mid \mathbf{u}^{(k)})$ majorizes $J(\mathbf{u})$ over $\mathcal{F}$ and is tight at $\mathbf{u}^{(k)}$. Dropping constant terms, the surrogate minimizer admits the projected-gradient form
\begin{equation} \label{eq:update_rule}
\mathbf{u}^{(k+1)}
=
\mathcal{P}_{\mathcal{F}}\!\left(
\mathbf{u}^{(k)}
-
L_k^{-1}\nabla J(\mathbf{u}^{(k)})
\right),
\end{equation}
where $\mathcal{P}_{\mathcal{F}}(\cdot)$ denotes the Euclidean projection onto $\mathcal{F}$. Since $\mathcal{F}$ is a 2D convex polygon bounded by linear hardware power constraints, $\mathcal{P}_{\mathcal{F}}(\cdot)$ can be evaluated in closed form via direct boundary clipping. Since the global Lipschitz constant is generally unavailable, we select $L_k$ by backtracking. Specifically, $L_k$ is increased by a factor $\tau$ until the majorization condition holds. This procedure maps each physical power update to a descent step on the ambiguity-aware BRB objective, ensuring monotonic decrease and convergence to a feasible stationary point under standard regularity conditions \cite{Sun_MM_2017}.

\subsection{Computational Complexity and Engineering Feasibility}

Computational latency is critical for practical 6G initial access. Classical beamforming methods based on semidefinite relaxation \cite{Luo_SDR_2010} require lifted covariance optimization and can scale as $\mathcal{O}(M^{4.5})$ with interior-point solvers. In contrast, the proposed backtracking MM framework optimizes only the two-dimensional vector $\mathbf{u}\in\mathbb{R}^2$. Since $\mathcal{F}$ is a two-dimensional polygon, the projection $\mathcal{P}_{\mathcal{F}}(\cdot)$ has constant per-iteration complexity. The remaining cost is dominated by evaluating the BRB objective and its gradient over the selected test points. 

Regarding scalability to large-scale arrays ($M \gg 1$), although increasing $M$ narrows the mainlobe, the number of dominant ambiguity peaks that dictate threshold breakdown remains small. By adopting a truncated test-point strategy that retains only the $L$ most prominent sidelobe peaks adjacent to the mainlobe (e.g., $L \le 6$), the dimension of the Barankin matrix $\boldsymbol{\mathcal{B}}(\theta) \in \mathbb{R}^{L \times L}$ is kept strictly low-dimensional and decoupled from $M$. Coupled with the diagonal regularization $\delta \mathbf{I}_L$ and truncated singular value decomposition (SVD), this guarantees numerical stability and constant matrix-inversion complexity even for massive array configurations. As shown in the simulations, the algorithm typically converges within 10 iterations, supporting low-latency power allocation within the THz channel coherence interval.

\begin{figure}[t]
\centering
\includegraphics[width=\columnwidth]{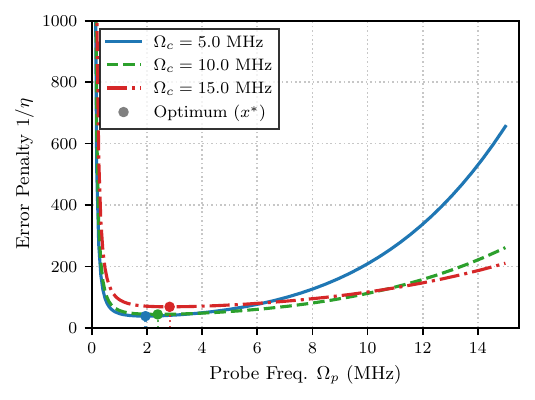}
\caption{Validation of the EIT readout trade-off among conversion gain, power broadening, and PSN.}
\label{fig:mechanism}
\end{figure}

\section{Simulation Results} \label{sec:simulation}

We evaluate the proposed framework using an RAQ-ULA with $M=10$, $\Delta d=0.5\lambda_c$, and $N=20$ snapshots. The hardware parameters are set as $P_{\mathrm{total}}=100~\mathrm{MHz}^2$, $\gamma=0.5$, $\Gamma_{\mathrm{EIT}}=6.0~\mathrm{MHz}$, $N_{\mathrm{bg}}=0.1$, and $c_{\mathrm{shot}}=0.2$. Truncated SVD with a tolerance of $10^{-10}$ is used to stabilize the Barankin matrix inversion. The empirical maximum likelihood estimator (MLE) is implemented on a dense spatial grid to reduce quantization errors.

Fig.~\ref{fig:mechanism} validates the physical trade-off characterized in Theorem~\ref{thm:optimal_power}. The estimation penalty $1/\eta$ exhibits a U-shaped dependence on the probe Rabi frequency $\Omega_p$, reflecting the competition among conversion-gain enhancement, power broadening, and PSN growth. The analytical optimum $\Omega_p^*$ closely matches the numerical minimizer, confirming that simply increasing the probe power does not necessarily improve quantum-array sensitivity. Instead, an interior operating point is required to balance the EIT response and the optical noise floor.

Fig.~\ref{fig:breakdown} compares stochastic performance bounds for DOA estimation. The stochastic CRB (Sto-CRB) provides a tight local bound in the high-SNR regime, but becomes overly optimistic at low SNR because it cannot capture finite-separation ambiguity errors. In contrast, the stochastic BRB (Sto-BRB) captures the nonlinear threshold breakdown around $4$ dB and better matches the empirical MLE behavior. This indicates that finite-separation likelihood ratios are essential for modeling ambiguity-dominated errors. The gap between the MLE and BRB in the transition region ($5$--$10$ dB) is mainly due to finite-snapshot effects, where random likelihood fluctuations may drive the estimator toward nearby ambiguity peaks. In the pure-noise region, the MLE approaches the saturation level $V_{\mathrm{arc}}$ induced by the bounded array manifold, which supports the proposed prior-calibrated BRB formulation.

Fig.~\ref{fig:convergence} shows the convergence behavior of the proposed backtracking MM algorithm from three initializations: A near the hardware boundary, B under an imbalanced EIT condition, and C from a low-power point. For all initializations, the Lipschitz surrogate in \eqref{eq:surrogate} yields a monotonic decrease in the BRB objective. The algorithm converges to a similar stationary value within 10 iterations, indicating stable descent despite the nonconvex BRB structure and supporting low-latency THz beam alignment.

\begin{figure}[t]
\centering
\includegraphics[width=\columnwidth]{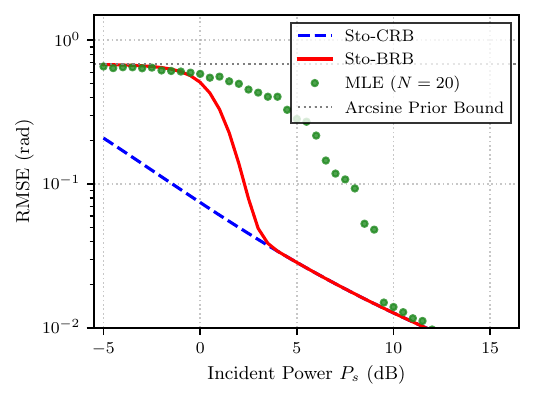}
\caption{Threshold-breakdown prediction using the stochastic BRB.}
\label{fig:breakdown}
\end{figure}

Fig.~\ref{fig:gain} compares the proposed optimization with heuristic laser-power allocation baselines. Equal-power allocation experiences threshold breakdown around $7$ dB because it does not explicitly balance power broadening and PSN growth. Maximum-power allocation can delay the breakdown by strengthening the received response, but excessive laser power broadens the atomic resonance linewidth and degrades the EIT readout condition. In contrast, the proposed MM-based allocation achieves a better trade-off by directly minimizing the prior-calibrated BRB, thereby balancing the PSN floor against the power-broadening penalty. To quantify the energy efficiency and physical operating state, the proposed MM algorithm converges to the optimal power allocation $(x^*, y^*) \approx (5.86, 94.14)~\mathrm{MHz}^2$, corresponding to laser Rabi frequencies $(\Omega_p^*, \Omega_c^*) \approx (2.42, 9.70)~\mathrm{MHz}$. In contrast, the maximum-power baseline saturates at the boundary $(x, y) = (\gamma P_{\mathrm{total}}/(1+\gamma), P_{\mathrm{total}}/(1+\gamma)) = (33.33, 66.67)~\mathrm{MHz}^2$ with $\Omega_p \approx 5.77~\mathrm{MHz}$. The proposed method restricts the probe intensity to an interior optimum, avoiding probe-induced power broadening and reducing the photon shot noise variance $\sigma^2$ by over $70\%$ compared to the saturated probe regime. Furthermore, when compared with the constrained SNR-maximization baseline derived from Theorem~\ref{thm:optimal_power}, which optimizes only the local Fisher curvature, the proposed BRB minimization directly suppresses spatial ambiguity lobes in the threshold region (around $0$--$5$~dB), achieving significantly lower estimation RMSE and extending the reliable operating range under severe link-budget deficit.

\begin{figure}[t]
\centering
\includegraphics[width=\columnwidth]{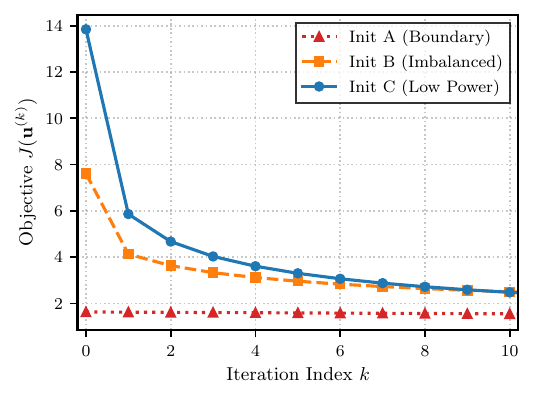}
\caption{Monotonic convergence of the backtracking MM algorithm.}
\label{fig:convergence}
\end{figure}

\section{Conclusion} \label{sec:conclusion}

This paper proposed a physics-aware power optimization framework for RAQ-ULAs in THz beam alignment. A closed-form stochastic multipoint BRB matrix with arcsine-prior calibration was derived to characterize finite-separation ambiguity errors and bounded-domain error saturation. By incorporating an EIT readout mapping, the framework links laser Rabi frequencies to the array covariance through quantum conversion gain, PSN, and power broadening. A backtracking MM algorithm was developed to minimize the prior-calibrated BRB under physical laser-power constraints. Simulations demonstrated improved threshold-breakdown prediction, monotonic objective decrease, and an enlarged reliable operating region under severe THz attenuation. Furthermore, while the current framework focuses on spatial sensing, integrating data transmission constraints to support Integrated Sensing and Communication (ISAC) paradigms \cite{Yuan_ISAC_2024} represents an important future direction. In addition, extending this physics-aware allocation methodology to distributed architectures, such as multi-agent trajectory optimization in UAV-assisted cell-free massive MIMO networks \cite{Liu_UAV_2025}, presents a promising avenue for complex network topologies.


\begin{thebibliography}{23}

\bibitem{Ai_2025_6G}
B.~Ai, Y.~Lu, Y.~Fang, D.~Niyato, R.~He, W.~Chen, J.~Zhang, G.~Ma, Y.~Niu, and Z.~Zhong, ``6G-enabled smart railways,'' \textit{Proc. IEEE}, vol.~113, no.~10, pp.~1192--1235, Oct.~2025.

\bibitem{Heath_2016}
R.~W. Heath, N.~Gonzalez-Prelcic, S.~Rangan, W.~Roh, and A.~M. Sayeed, ``An overview of signal processing techniques for millimeter wave MIMO systems,'' \textit{IEEE J. Sel. Topics Signal Process.}, vol.~10, no.~3, pp.~436--453, Apr.~2016.

\bibitem{Gong_RAQR_2026}
T.~Gong, J.~Sun, C.~Yuen, G.~Hu, Y.~Zhao, Y.~L.~Guan, C.~M.~S. See, M.~Debbah, and L.~Hanzo, ``Rydberg atomic quantum receivers for classical wireless communications and sensing: Their models and performance,'' \textit{IEEE Trans. Commun.}, vol.~74, pp.~7759--7778, Apr.~2026.

\bibitem{Liu_2022}
Z.-K.~Liu, L.-H.~Zhang, B.~Liu, Z.-Y.~Zhang, G.-C.~Guo, D.-S.~Ding, and B.-S.~Shi, ``Deep learning enhanced Rydberg multifrequency microwave recognition,'' \textit{Nat. Commun.}, vol.~13, Art.~no.~1997, Apr.~2022.

\bibitem{Richardson_2025}
D.~Richardson, J.~Dee, B.~N.~Kayim, B.~C.~Sawyer, R.~Wyllie, R.~T.~Lee, and R.~S.~Westafer, ``Study of angle of arrival estimation with linear arrays of simulated Rydberg atom receivers,'' \textit{APL Quantum}, vol.~2, no.~1, Art.~no.~016123, Feb.~2025.

\bibitem{Meyer_2021}
D.~H. Meyer, C.~O'Brien, D.~P. Fahey, K.~C. Cox, and P.~D. Kunz, ``Optimal atomic quantum sensing using electromagnetically-induced-transparency readout,'' \textit{Phys. Rev. A}, vol.~104, Art.~no.~043103, Oct.~2021.

\bibitem{Fleischhauer_2005}
M.~Fleischhauer, A.~Imamoglu, and J.~P. Marangos, ``Electromagnetically induced transparency: Optics in coherent media,'' \textit{Rev. Mod. Phys.}, vol.~77, no.~2, pp.~633--673, Jul.~2005.

\bibitem{Krim_1996}
H.~Krim and M.~Viberg, ``Two decades of array signal processing research: The parametric approach,'' \textit{IEEE Signal Process. Mag.}, vol.~13, no.~4, pp.~67--94, Jul.~1996.

\bibitem{VanTrees_2002}
H.~L. Van Trees, \textit{Optimum Array Processing: Part IV of Detection, Estimation, and Modulation Theory}. New York, NY, USA: Wiley-Interscience, 2002.

\bibitem{Wang_BeamSquint_2019}
B.~Wang, M.~Jian, F.~Gao, G.~Y. Li, and H.~Lin, ``Beam squint and channel estimation for wideband mmWave massive MIMO-OFDM systems,'' \textit{IEEE Trans. Signal Process.}, vol.~67, no.~23, pp.~5893--5908, Dec.~2019.

\bibitem{Xu_NearField_2024}
B.~Xu, J.~Zhang, H.~Du, Z.~Wang, Y.~Liu, D.~Niyato, B.~Ai, and K.~B. Letaief, ``Resource allocation for near-field communications: Fundamentals, tools, and outlooks,'' \textit{IEEE Wireless Commun.}, vol.~31, no.~5, pp.~42--50, Oct.~2024.

\bibitem{xu2025channelestimation}
B.~Xu, J.~Zhang, Z.~Chen, B.~Cheng, Z.~Liu, Y.-C.~Wu, and B.~Ai, ``Channel estimation for Rydberg atomic receivers,'' \textit{IEEE Wireless Commun. Lett.}, vol.~14, no.~9, pp.~2957--2961, Sep.~2025.

\bibitem{Meyer_Wideband_2020}
D.~H. Meyer, Z.~A. Castillo, K.~C. Cox, and P.~D. Kunz, ``Assessment of Rydberg atoms for wideband electric field sensing,'' \textit{J. Phys. B, At. Mol. Opt. Phys.}, vol.~53, no.~3, Art.~no.~034001, Jan.~2020.

\bibitem{Stoica_1990}
P.~Stoica and A.~Nehorai, ``Performance study of conditional and unconditional direction-of-arrival estimation,'' \textit{IEEE Trans. Acoust., Speech, Signal Process.}, vol.~38, no.~10, pp.~1783--1795, Oct.~1990.

\bibitem{Chaumette_2008}
E.~Chaumette, J.~Galy, A.~Quinlan, and P.~Larzabal, ``A new Barankin bound approximation for the prediction of the threshold region performance of maximum-likelihood estimators,'' \textit{IEEE Trans. Signal Process.}, vol.~56, no.~11, pp.~5319--5333, Nov.~2008.

\bibitem{Liu_NearField_2024}
M.~Liu, Y.~Zhang, Y.~Jin, K.~Zhi, and C.~Pan, ``Towards near-field communications for 6G: Challenges and opportunities,'' \textit{ZTE Commun.}, vol.~22, no.~1, pp.~3--15, Mar.~2024.

\bibitem{Lu_NearField_2022}
H.~Lu and Y.~Zeng, ``Communicating with extremely large-scale array/surface: Unified modeling and performance analysis,'' \textit{IEEE Trans. Wireless Commun.}, vol.~21, no.~6, pp.~4039--4053, Jun.~2022.

\begin{figure}[t]
\centering
\includegraphics[width=\columnwidth]{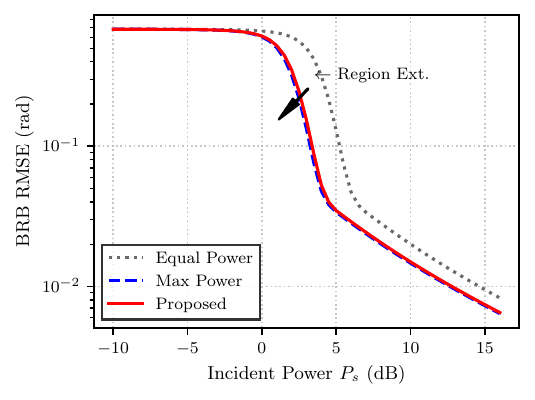}
\caption{Extension of the reliable operating region compared with heuristic power allocations.}
\label{fig:gain}
\end{figure}

\bibitem{Liu_CellFree_2024}
Z.~Liu, J.~Zhang, Z.~Liu, H.~Du, Z.~Wang, D.~Niyato, M.~Guizani, and B.~Ai, ``Cell-free XL-MIMO meets multi-agent reinforcement learning: Architectures, challenges, and future directions,'' \textit{IEEE Wireless Commun.}, vol.~31, no.~4, pp.~155--162, Aug.~2024.

\bibitem{Renaux_Prior_2008}
A.~Renaux, P.~Forster, P.~Larzabal, C.~D. Richmond, and A.~Nehorai, ``A fresh look at the Bayesian bounds of the Weiss-Weinstein family,'' \textit{IEEE Trans. Signal Process.}, vol.~56, no.~11, pp.~5334--5352, Nov.~2008.

\bibitem{Sun_MM_2017}
Y.~Sun, P.~Babu, and D.~P. Palomar, ``Majorization-minimization algorithms in signal processing, communications, and machine learning,'' \textit{IEEE Trans. Signal Process.}, vol.~65, no.~3, pp.~794--816, Feb.~2017.

\bibitem{Luo_SDR_2010}
Z.-Q.~Luo, W.-K.~Ma, A.~M.-C. So, Y.~Ye, and S.~Zhang, ``Semidefinite relaxation of quadratic optimization problems,'' \textit{IEEE Signal Process. Mag.}, vol.~27, no.~3, pp.~20--34, May~2010.

\bibitem{Yuan_ISAC_2024}
J.~H. Yuan, Z.~S. Fei, and Z.~Q. Wei, ``Integrated sensing and communication (ISAC) technologies for future wireless communication,'' \textit{ZTE Commun.}, vol.~22, no.~3, pp.~1--3, Sep.~2024.

\bibitem{Liu_UAV_2025}
Z.~Liu, J.~Zhang, Y.~Zeng, and B.~Ai, ``Energy-efficient multi-agent reinforcement learning for UAV trajectory optimization in cell-free massive MIMO networks,'' \textit{IEEE Trans. Wireless Commun.}, vol.~24, no.~7, pp.~5917--5930, Jul.~2025.

\end{thebibliography}
\end{document}